\documentclass[11pt]{article}
\usepackage[hmargin=3.5cm,vmargin=3cm]{geometry}
\usepackage{amsmath}
\usepackage{amssymb}
\usepackage{amsthm}
\usepackage[mathscr]{eucal}
\usepackage{enumitem}
\usepackage[pdftex]{color}
\definecolor{darkblue}{rgb}{0.3,0.3,0.7}
\usepackage[pdftitle={Sticky continuous processes have consistent price systems},pdfauthor={Christian Bender, Mikko S. Pakkanen, and Hasanjan Sayit},colorlinks=TRUE,allcolors=darkblue]{hyperref}

\newtheorem{theorem}{Theorem}[section]
\newtheorem{lem}[theorem]{Lemma}

\theoremstyle{definition}

\newtheorem{defn}[theorem]{Definition}
\newtheorem{ex}[theorem]{Example}
\theoremstyle{remark}
\newtheorem{rem}[theorem]{Remark}
\numberwithin{equation}{section}
\newcommand{\ud}{\mathrm{d}}

\newcommand{\eqdefl}{\mathrel{\mathop:}=}

\newcommand{\Q}{\mathbb{Q}}
\newcommand{\R}{\mathbb{R}}

\DeclareMathOperator{\conv}{conv}
\DeclareMathOperator{\supp}{supp}
\DeclareMathOperator{\ri}{ri}
\DeclareMathOperator{\interior}{int}

\newcommand{\F}{\mathscr{F}}

\begin{document}
\title{Sticky continuous processes have \\ consistent price systems}
\date{April 17, 2014}

\author{Christian Bender\thanks{Department of Mathematics, Saarland University, Postfach 151150, D-66041 Saarbr\"ucken, Germany. E-mail: \href{mailto:bender@math.uni-sb.de}{\nolinkurl{bender@math.uni-sb.de}}},
Mikko S. Pakkanen\thanks{
CREATES and Department of Economics and Business, Aarhus University, Fuglesangs All\'e 4, DK-8210 Aarhus V, Denmark. E-mail: \href{mailto:mpakkanen@creates.au.dk}{\nolinkurl{mpakkanen@creates.au.dk}}},
Hasanjan Sayit\thanks{Department of Mathematical Sciences, Durham University, South Road, Durham DH1 3LE, UK. E-mail: \href{mailto:hasanjan.sayit@durham.ac.uk}{\nolinkurl{hasanjan.sayit@durham.ac.uk}}}
}

\maketitle

\begin{abstract}
Under proportional transaction costs, a price process is said to have a consistent price system, if there is a semimartingale with an equivalent martingale measure that 
evolves within the bid-ask spread. We show that a continuous, multi-asset price process has a consistent price system, under arbitrarily small proportional transaction costs, if it satisfies a natural multi-dimensional generalization of the stickiness condition introduced by Guasoni \cite{Gua}.

\vspace*{1ex}


\vspace*{1ex}


\vspace*{1ex}


\end{abstract}

\section{Introduction}

In an asset pricing model that imposes proportional transaction costs on trading, a \emph{consistent price system} (CPS) is a shadow price process with an equivalent martingale measure that evolves within the bid-ask spread implied by transaction costs, a concept that dates back to the paper by Jouini and Kallal \cite{JK}. 
It is clear that trading under transaction costs cannot be more profitable than trading at the prices given by the CPS \emph{without frictions}. Thus, questions concerning, for example, the \emph{no arbitrage} (NA) condition under transaction costs can be answered using frictionless theory. In fact, such a transaction cost model satisfies NA if there exists a CPS --- any arbitrage would also be an arbitrage with respect to the CPS without frictions, a contradiction.

The question whether a continuous-time price process has a CPS under proportional transaction costs has been studied by several authors in the case where the process 
is \emph{continuous}, the case which we focus on in this note. Guasoni et al.\ \cite[Theorem 1.2]{Gua1} showed that a multi-asset price process has a CPS under arbitrarily
 small transaction costs if it has the so-called \emph{conditional full support} (CFS) property 
 (see Remark \ref{CFS}, below).  
Kabanov and Stricker \cite[Theorem 1]{KS} established a similar existence result under the assumptions that the price process does not admit arbitrage 
opportunities with simple
 trading strategies and is \emph{sticky} (we will elaborate on this property shortly). However, the conditions of these two existence results are not necesssary for the existence of CPSs (see Example \ref{mono}, below).

In fact, in the single-asset case, Guasoni et al.\ \cite[Theorem 2]{Gua2} have established that a CPS exists under arbitrarily small transaction costs if and only if
 the price process satisfies the NA condition under arbitrarily small transaction costs (a \emph{fundamental theorem of asset pricing}). Earlier, Guasoni
 \cite{Gua} showed that a price process satisfies the NA condition under arbitrarily small transaction costs, if it is merely sticky.
Informally, stickiness entails that the process remains in any neighborhood of its current value with a positive conditional probability
 (see Definition \ref{stickiness}, below, for a rigorous formulation). Combining the results in \cite{Gua,Gua2}, it follows in the single-asset case that a continuous process has a CPS under arbitrarily small transaction costs if it is sticky.

The purpose of this note is to show that also in the multi-asset case any sticky continuous process has a CPS under arbitrarily small transaction costs. In particular, we give a direct proof of this statement, that is, even in the single-asset case where the fundamental theorem of Guasoni et al. \cite[Theorem 2]{Gua2} is available, we do not rely on such a result. Instead, our proof is based on the arguments used in \cite{Gua1,KS}, but we modify them in a novel way (see Remark \ref{difference}, below, for a discussion).
In this way we can carry out the construction of a CPS under the stickiness condition alone, which is weaker and easier to check than the assumptions in 
\cite{Gua1,KS}.

\section{Preliminaries and main result}

Our probabilistic setup is as follows.
Let $T \in \R_+ \eqdefl (0,\infty)$ be a fixed time horizon and $\big(\Omega,\F,(\F_t)_{t \in [0,T]},\mathbf{P}\big)$ a complete filtered probability space, such that the filtration $(\F_t)_{t \in [0,T]}$ satisfies $\F_T = \F$ along with the usual conditions of right-continuity and completeness. 
We say that an ``$\mathscr{F}$-measurable'' property $\mathscr{P}$ holds \emph{almost surely on} $A \in \mathscr{F}$ if $\mathbf{P}[\{\mathscr{P}\}\cap A] = \mathbf{P}[A]$. 
Furthermore, let $d \in \mathbb{N} \eqdefl \{1,2,\ldots\}$ and let $(S_t)_{t \in [0,T]}$, where $S_t \eqdefl (S^1_t,\ldots,S^d_t)$, be an adapted $d$-dimensional process on $\big(\Omega,\F,(\F_t)_{t \in [0,T]},\mathbf{P}\big)$ with continuous paths and values in $\R^d_+$. In terms of economic interpretation, $S$ describes the evolution of the prices of $d$ risky securities in terms of a money market. 

\begin{defn}
An adapted $d$-dimensional stochastic process $(M_t)_{t \in [0,T]}$, where $M_t \eqdefl (M^1_t,\ldots,M^d_t)$, is called an $\varepsilon$-\emph{consistent price system} ($\varepsilon$-CPS) for $S$ if for any $i \in \{1,\ldots,d\}$ and $t \in [0,T]$,
\begin{equation*}
\frac{S^i_t}{1+\varepsilon}\leq M^i_t\leq  (1+\varepsilon)S^i_t \quad \textrm{ a.s.,}
\end{equation*}
and if there is a probability measure $\mathbf{Q}$ on $(\Omega,\F)$ such that $\mathbf{Q} \sim \mathbf{P}$ and that $M$ is a $\mathbf{Q}$-martingale.
\end{defn}

If $S$ has an $\varepsilon$-CPS, then $S$, as a price process, is free of arbitrage and free lunches with vanishing risk with $\varepsilon$-sized proportional 
transaction costs (which follows, e.g., by combining Corollary 1.2 of \cite{DS94} and Lemma 2.1 of \cite{Gua}). We refer the reader to \cite{Gua1,Gua2} for details on 
 how the value of a portfolio is defined under proportional transactions costs in continuous time. 

To formulate our main result, we recall the notion of \emph{stickiness} --- here in a multidimensional form --- that was initially introduced by Guasoni \cite{Gua}.
To this end, we use the norm $\|x\| \eqdefl \max(|x_1|,\ldots,|x_d|)$ for any $x = (x_1,\ldots,x_d) \in \R^d$, and write for any stopping time $\tau \in [0,T]$,
\begin{equation*}
S^\star_\tau \eqdefl \sup_{u \in [\tau,T]} \|S_u -S_\tau\|.
\end{equation*}

\begin{defn}\label{stickiness}
The process $S$ is \emph{sticky}, if for any $t \in [0,T)$ and $\delta >0$,
\begin{equation*}
\mathbf{P}[ S^\star_t<\delta | \F_t] >0 \textrm{ a.s.}
\end{equation*}
\end{defn}

\begin{rem}
We stress that $t$ in Definition \ref{stickiness} is \emph{non-random}.
This definition of stickiness is, however, equivalent to the concept of \emph{joint stickiness} introduced in \cite[Definition 2]{SV}, by Lemma \ref{strong-stickiness}, below. When $d=1$, Definition \ref{stickiness} is also equivalent to Guasoni's original one-dimensional definition of stickiness \cite[Definition 2.2]{Gua} by Proposition 1 of \cite{BS} and Lemma \ref{strong-stickiness}.
\end{rem}

\begin{rem}\label{CFS}
Recall that the process $S$ has \emph{conditional full support} (CFS) if for any $t \in [0,T)$, $\delta>0$, and for any continuous function $f : [0,T] \rightarrow \R^d$ such that $f(0)=0$, it holds that
\begin{equation*}
\mathbf{P}\bigg[ \sup_{u \in [t,T]} \|S_u - S_t -f(u-t)\| < \delta \bigg| \mathscr{F}_t \bigg] >0
\end{equation*}
a.s.\ on $\{ f(u-t)+S_t \in \R^d_+ \textrm{ for all $u \in [t,T]$}\}$. Intuitively, the process $(S_u)_{u \in [t,T]}$ then ``sticks'' to any random function of the form $f(\cdot - t)+S_t$ with positive conditional probability. The CFS property is clearly a stronger requirement than stickiness, which entails that $S$ ``sticks'' merely to its current value with positive conditional probability (that is, $f=0$). The CFS property has been recently established for a wide selection of processes, see \cite{C,GSvZ,Gua1,GR,HPR,P2010,P2011}. 
In particular, many Gaussian processes, including \emph{fractional Brownian motion}, have CFS.

It is worth pointing out that stickiness is obviously preserved under composition of the process $S$ with a continuous function, which is not true in general for the CFS property. Due to this observation one can easily 
construct an abundance of sticky processes by composing, e.g., a process having CFS with a continuous function.
\end{rem}

We can now state our main result, deferring its proof until Section \ref{sec:proof}.

\begin{theorem}[Consistent price systems]\label{CPS-existence}
If $S$ is sticky, then it admits an $\varepsilon$-CPS for any $\varepsilon>0$.
\end{theorem}

Curiously, one of the implications of Theorem \ref{CPS-existence} is that there exist examples of a process with monotonic, smooth paths that can be approximated, with arbitrary precision, by a process that transforms into a martingale under an equivalent probability measure --- which is perhaps a bit counterintuitive from a probabilistic point of view. 

\begin{ex}[Strictly increasing process]\label{mono}
Suppose that $(B_t)_{t \in [0,T]}$ is a standard (one-dimensional) Brownian motion and $s_0>0$ is a constant. Let us consider the process
\begin{equation*}
S_t \eqdefl s_0 + \int_0^t |B_s| \ud s, \quad t \in [0,T].
\end{equation*}
Clearly, $S$ has strictly increasing, continuously differentiable paths. Thus, neither the CFS property  
(see Remark \ref{CFS}) 
nor the criterion of Kabanov and Stricker \cite{KS} is in force.
However, using the independence and stationarity of increments of $B$, and their full support property \cite[Corollary VIII.2.3]{RY}, it is straightforward to show that $S$ is sticky. Thus, by Theorem \ref{CPS-existence}, there exists for any $\varepsilon>0$ a probability measure $\mathbf{Q} \sim \mathbf{P}$ and a $\mathbf{Q}$-martingale $M$ such that $|M_t/S_t -1 | \leq \varepsilon$ a.s.\ for all $t \in [0,T]$. Note that the monotonicity and differentiability of $S$ is preserved under $\mathbf{Q}$ since $\mathbf{Q} \sim \mathbf{P}$.
\end{ex}

\section{Proof of Theorem \ref{CPS-existence}}\label{sec:proof}

We will first introduce some additional notation and concepts that are crucial in what follows.
Suppose that $E \subset \R^d$. Recall that the \emph{convex hull} of $E$, denoted by $\conv E$, is the smallest convex subset of $\R^d$ that contains $E$. Moreover, the \emph{relative interior} of $E$, denoted by $\ri E$, is the interior of $E$ in the relative topology of the smallest affine subset of $\R^d$ that contains $E$ (the so-called \emph{affine hull} of $E$). We use $\interior E$ for the (ordinary) interior of $E$.  The \emph{support} of a probability measure $\mu$ on $\R^d$, denoted by $\supp \mu$ is the smallest closed set $E \subset \R^d$ such that $\mu(E) = 1$. The (regular) \emph{conditional law} of a random vector $X$ in $\R^d$, defined on $(\Omega,\mathscr{F},\mathbf{P})$, with respect to a $\sigma$-algebra $\mathscr{G} \subset \mathscr{F}$ will be denoted by $\mathscr{L}(X|\mathscr{G})$. We use the standard convention that $\inf \varnothing = \infty$.

As a preparation for the proof of Theorem \ref{CPS-existence}, we show that stickiness implies a (seemingly) stronger property, obtained by replacing the deterministic time $t$ and radius $\delta$ of Definition \ref{stickiness} with a stopping time $\tau$ and an $\mathscr{F}_\tau$-measurable random radius $\eta$, respectively. This result is analogous to Proposition 2.9 of \cite{Gua1}, which says that the CFS property implies the strong CFS property, and its proof is, in fact, carried out in a similar fashion (cf.\ also Lemma 2.2 of \cite{MMS}).

\begin{lem}\label{strong-stickiness}
If $S$ is sticky, then for any stopping time $\tau \in [0,T]$ and $\mathscr{F}_\tau$-measurable random variable $\eta\geq 0$, 
\begin{equation*}
\mathbf{P}[ S^\star_\tau < \eta |\mathscr{F}_\tau]>0 \quad \textrm{a.s. on $\{\eta > 0 \}$.}
\end{equation*}
\end{lem}

\begin{proof}
To exclude a triviality, we assume that $\mathbf{P}[\eta>0]>0$. 
We use a contrapositive argument and suppose that the assertion is not true, that is, there is a stopping time $\tau\in [0,T]$, an $\mathscr{F}_\tau$-measurable random variable $\eta$, and $A \in\mathscr{F}_\tau$ such that $A \subset \{ \eta>0 \}$, $\mathbf{P}[A]>0$, and
\begin{equation}\label{contraposition}
\mathbf{P}[ S^\star_\tau \geq \eta |\mathscr{F}_\tau]=1 \quad \textrm{on $A$.}
\end{equation}
Note that, necessarily, $A \subset \{\tau < T\}$. Since $A \subset \{ \eta>0 \}$, we may write $A = \bigcup_{n\in \mathbb{N}} A'_n$, where $A'_n \eqdefl A \cap \{ \eta > 2/n \} \in \mathscr{F}_\tau$. The property $\mathbf{P}[A]>0$ implies then that $\mathbf{P}[A'_{n_0}]>0$ for some $n_0 \in \mathbb{N}$. Further, by the continuity of $S$, we have $A'_{n_0} = \bigcup_{q \in [0,T)\cap \Q} A''_q$, where
\begin{equation*}
A''_q \eqdefl A'_{n_0} \cap \{\tau < q \} \cap \bigg\{\sup_{t \in [\tau,q]} \|S_t -S_\tau \|< \frac{1}{{n_0}}\bigg\} \in \mathscr{F}_q.
\end{equation*}
Thus, $\mathbf{P}[A'_{n_0}]>0$ implies that $\mathbf{P}[A''_{q_0}]>0$ for some $q_0  \in [0,T)\cap \Q$. 

The property \eqref{contraposition} implies that $\mathbf{1}_{ \{S^\star_\tau \geq \eta \}} = 1$ a.s.\ on $A$. Thus, recalling that $A''_{q_0} \subset A$,
\begin{equation}\label{lower}
\mathbf{P}[A''_{q_0}] = \mathbf{E}[\mathbf{1}_{A''_{q_0}}\mathbf{1}_{ \{S^\star_\tau \geq \eta \}}]=\mathbf{P}[A''_{q_0} \cap \{S^\star_\tau \geq \eta \}].
\end{equation}
Since we have the inclusion
\begin{equation*}
A''_{q_0} \cap \{ S^\star_\tau \geq \eta \} \subset A''_{q_0} \cap \{ S^\star_{q_0} \geq 1/{n_0} \},
\end{equation*}
it follows that
\begin{equation}\label{upper}
\mathbf{P}[A''_{q_0}\cap\{S^\star_\tau \geq \eta \}] \leq \mathbf{P}[A''_{q_0} \cap \{ S^\star_{q_0} \geq 1/{n_0} \}] \leq \mathbf{P}[A''_{q_0}].
\end{equation}
Combining \eqref{lower} and \eqref{upper}, we obtain
\begin{equation*}
\mathbf{P}[S^\star_{q_0} \geq 1/{n_0} | \mathscr{F}_{q_0}] = 1 \quad \textrm{a.s.\ on $A''_{q_0}$,}
\end{equation*}
which implies that $S$ is not sticky.
\end{proof}

Similarly to the paper by Kabanov and Stricker \cite{KS}, the key ingredient in the proof of our Theorem \ref{CPS-existence} is the following result (Lemma 1 in the revised version of \cite{KS}) that gives sufficient conditions for the existence of an equivalent martingale measure for a discrete parameter process with infinite time horizon (see also Remark \ref{alternative}, below).

\begin{lem}[Equivalent martingale measure] \label{keylem}
Let $(\mathscr{G}_n)_{n \geq 0}$ be a filtration on $(\Omega,\mathscr{F},\mathbf{P})$ and let $(X_n)_{n \geq 0}$ be a discrete-parameter, $d$-dimensional process adapted to $(\mathscr{G}_n)_{n \geq 0}$. Further, write $\xi_n \eqdefl X_n - X_{n-1}$ for any $n \in \mathbb{N}$.
If
\begin{enumerate}[label=(\roman*),ref=\roman*]
\item\label{c-1} $0\in \ri \conv \supp \mathscr{L}(\xi_{n}|\mathscr{G}_{n-1})$ a.s.\ for any $n \in \mathbb{N}$,
\item\label{c-2} $\mathbf{1}_{\{ \xi_{n}=0\}}\uparrow 1$ a.s.\ when $n \uparrow \infty$,
\item\label{c-3} $\mathbf{P}[\xi_{n}=0|\mathscr{G}_{n-1}]>0$ a.s.\ on $\{ \xi_{n-1}\neq0\}$ for any $n \in \mathbb{N}$,
\end{enumerate}
then there exists a probability measure $\mathbf{Q} \sim \mathbf{P}$ such that $X$ is a $\mathbf{Q}$-martingale and bounded in $L^2(\mathbf{Q})$.
\end{lem}

\begin{rem}
Kabanov and Stricker \cite{KS} write the condition \eqref{c-1} in a different form that requires the process $(X_n)_{n=1}^N$ to be free of arbitrage for any time horizon $N \in \mathbb{N}$.
 For our purposes, the geometric formulation \eqref{c-1} will be more convenient. Both formulations are equivalent by Theorem A* in Chapter V of \cite{Sh}. 
\end{rem}

\begin{proof}[Proof of Theorem \ref{CPS-existence}]
Fix $\varepsilon>0$.
Let us define an increasing sequence $(\tau_n)_{n\geq0}$ of stopping times by setting $\tau_0 \eqdefl 0$ and for any $n \in \mathbb{N}$,
\begin{equation*}
\tau_{n}\eqdefl\inf\bigg\{t\geq \tau_{n-1}: \frac{S^i_t}{S^i_{\tau_{n-1}}} \notin \bigg(\frac{1}{1+\varepsilon},1+\varepsilon\bigg) \textrm{ for some $i = 1,\ldots,d$}\bigg\}\wedge T.
\end{equation*}
As argued in \cite[p.\ 509]{Gua1}, the continuity of the process $S$ ensures that $\mathbf{1}_{\{\tau_n = T\}} \uparrow 1$ a.s.\ when $n\uparrow \infty$. 

Following the proof of Theorem 1 in \cite{KS}, we approximate the skeleton $(S_{\tau_n})_{n \geq 0}$ of the process $S$ by a carefully constructed discrete-parameter process $(X_n)_{n \geq 0}$ that will satisfy the conditions of Lemma \ref{keylem} with $(\mathscr{G}_{n})_{n \geq 0} \eqdefl (\mathscr{F}_{\tau_n})_{n \geq 0}$. 
Deviating from the construction of $X$ in  \cite{Gua1,KS} we first decompose the event 
$\{ \tau_n =T,\tau_{n-1}< T\}$ into $(2d+1)$ disjoint subevents such that each of them has positive 
$\mathscr{F}_{\tau_{n-1}}$-conditional probability, if $S$ moves after $\tau_{n-1}$.
On one of these subevents, we set $X_n - X_{n-1} \eqdefl 0$ to ensure that condition \eqref{c-3} of Lemma \ref{keylem} 
holds.
 On the remaining $2d$ 
subevents, $X_n-X_{n-1}$ will be assigned values given by the $d$ standard basis vectors of $\R^d$ and their opposites, scaled by a small positive factor so that $X$ remains ``close'' to the skeleton $(S_{\tau_n})_{n \geq 0}$.
In this way we can guarantee that $X$ satisfies 
the support condition \eqref{c-1} of Lemma \ref{keylem}.

We proceed now with the rigorous construction of the discrete-parameter process $X$.
For any stopping time $\tau \in [0,T]$, we introduce
\begin{equation}\label{support-maximum}
\overline{S}_\tau  \eqdefl \sup \supp \mathscr{L}(S^\star_\tau| \mathscr{F}_\tau).
\end{equation}
Observe that the random variable $\overline{S}_\tau$ assumes values in $[0,\infty]$ and is $\mathscr{F}_\tau$-measurable since
we can alternatively represent it as
\begin{equation*}
\overline{S}_\tau = \sup_{s \in \Q_+} s \mathbf{1}_{\{ \mathbf{P}[S^\star_\tau \geq s | \mathscr{F}_\tau] > 0\}},
\end{equation*}
where $\{ \mathbf{P}[S^\star_\tau \geq s | \mathscr{F}_\tau] > 0\} \in \mathscr{F}_\tau$ and $\Q_+ \eqdefl \R_+ \cap \Q$.
We set
\begin{equation*}
\delta_{n-1} \eqdefl \min\bigg(\frac{\varepsilon}{1+\varepsilon}S^1_{\tau_{n-1}},\ldots,\frac{\varepsilon}{1+\varepsilon}S^d_{\tau_{n-1}},\overline{S}_{\tau_{n-1}}\bigg),
\end{equation*}
which is an $\mathscr{F}_{\tau_{n-1}}$-measurable (finite) random variable, and 
\begin{equation}\label{cdef}
C_{n}^{i} \eqdefl \bigg\{ S^\star_{\tau_{n-1}} \in \bigg[\frac{i-1}{2d+1}\delta_{n-1},\frac{i}{2d+1}\delta_{n-1}\bigg) \bigg\}\cap \big\{\overline{S}_{\tau_{n-1}}>0 \big\}\subset \{\tau_n = T,\tau_{n-1}<T\},
\end{equation}
with $n\in \mathbb{N}$ and $i \in \{1,\ldots,2d+1\}$. 
 Below, in Lemma \ref{prob}, we show these events have positive $\mathscr{F}_{\tau_{n-1}}$-conditional probabilities, as a consequence of the stickiness of the process $S$. We define $X$ to be the $(\mathscr{F}_{\tau_n})_{n \geq 0}$-adapted process
\begin{equation*}
X_n \eqdefl S_0 + \sum_{k=1}^n \xi_k, \quad n \geq 0,
\end{equation*}
where
\begin{equation}\label{increment-def}
\xi_n \eqdefl (S_{\tau_{n}} - S_{\tau_{n-1}})\mathbf{1}_{\{ \tau_n <T\}} + \sum_{i=1}^{d}  \delta_{n-1}(\mathbf{1}_{ C_{n}^{2i}}-\mathbf{1}_{ C_{n}^{2i+1}})e_i \quad \textrm{for any $n \in \mathbb{N}$,}
\end{equation}
using $e_i$ to denote the $i$-th standard basis vector of $\R^d$. 

For any $n \in \mathbb{N}$, the second term in \eqref{increment-def} ensures that $\xi_n = \delta_{n-1} e_i$ on $C^{2i}_n$ and $\xi_n = -\delta_{n-1} e_i$ on $C^{2i+1}_n$ with $i \in \{1,\ldots,d\}$ and, moreover, $\xi_n = 0$ on $C^{1}_n$.
To verify condition \eqref{c-1} of Lemma \ref{keylem}, note that
\begin{equation}\label{atom}
\mathrm{supp}\,\mathscr{L}(\xi_n|\mathscr{F}_{\tau_{n-1}}) = \{0\} \quad \textrm{a.s.\ on $\big\{\overline{S}_{\tau_{n-1}}=0\big\}$},
\end{equation}
whereas
\begin{equation}\label{all-directions}
\mathrm{supp}\,\mathscr{L}(\xi_n|\mathscr{F}_{\tau_{n-1}}) \supset \{ \pm\delta_{n-1}e_i : i = 1,\ldots,d\} \quad \textrm{a.s.\ on $\big\{\overline{S}_{\tau_{n-1}}>0\big\}$}
\end{equation}
by the construction of $\xi_n$ on the sets $C^i_n$ and Lemma \ref{prob} below. Thus, it follows that
\begin{equation*}
0 \in \mathrm{ri}\,\mathrm{conv}\,\mathrm{supp}\,\mathscr{L}(\xi_n|\mathscr{F}_{\tau_{n-1}})\quad \textrm{a.s.}
\end{equation*}
The condition \eqref{c-2} of Lemma \ref{keylem} is satisfied since $\{\xi_{n+1} = 0 \} \supset \{\xi_n = 0 \} \supset \{\tau_{n-1}=T\}$ and
\begin{equation*}
\mathbf{1}_{\{\xi_n=0\}} \geq \mathbf{1}_{\{\tau_{n-1}=T\}}\uparrow 1 \quad \textrm{a.s.\ when $n \uparrow \infty$.}
\end{equation*}
Finally, condition \eqref{c-3} of Lemma \ref{keylem} follows from \eqref{atom} and the inequalities
\begin{equation}\label{zero-possible}
\mathbf{P}[\xi_n = 0 |\mathscr{F}_{\tau_{n-1}}]\geq \mathbf{P}[C^1_{n}|\mathscr{F}_{\tau_{n-1}}]>0 \quad \textrm{a.s.\ on $\big\{\overline{S}_{\tau_{n-1}}>0\big\}$,}
\end{equation}
the latter of which comes from Lemma \ref{prob} below.

The construction of a CPS can now be carried out using a standard method (cf.\ \cite{Gua1,KS}).
By Lemma \ref{keylem}, there exists a probability measure $\mathbf{Q} \sim \mathbf{P}$ such that $X$ is a uniformly integrable $\mathbf{Q}$-martingale, thus closable at $\infty$. We may then define a $d$-dimensional continuous-parameter $\mathbf{Q}$-martingale $M$ by setting $M_t \eqdefl \mathbf{E}_{\mathbf{Q}}[X_\infty|\mathscr{F}_t]$ for any $t \in [0,T]$. By the optional sampling theorem, $ M_{\tau_n}=X_n$ a.s.\ for any $n \geq 0$. For the remainder of the proof, let us fix arbitrary $i \in \{1,\ldots,d\}$ and $t \in [0,T]$. The processes $X$ and $M$ satisfy, by construction, for any $n \geq 0 $,
\begin{equation}\label{eps-bound1}
\frac{1}{(1+\varepsilon)^2} \leq \frac{X^i_n}{S^i_{\tau_{n}}} = \frac{M^i_{\tau_n}}{S^i_{\tau_{n}}} \leq (1+ \varepsilon)^2 \quad \textrm{a.s.}
\end{equation}
It is shown in \cite[p.\ 510]{Gua1} that
\begin{equation}\label{eps-bound2}
\frac{1}{(1+\varepsilon)^2} \leq \frac{S^i_{\rho_t}}{S^i_{t}} \leq (1+ \varepsilon)^2\quad \textrm{a.s.,}
\end{equation}
where $\rho_t \eqdefl \min \{ \tau_n : \tau_n > t\}$, which is a stopping time. It follows then from \eqref{eps-bound1} and \eqref{eps-bound2} that
\begin{equation*}
\frac{1}{(1+\varepsilon)^4} \leq \frac{M^i_{\rho_t}}{S^i_t} = \frac{M^i_{\rho_t}}{S^i_{\rho_t}}\frac{S^i_{\rho_t}}{S^i_{t}} \leq (1+ \varepsilon)^4\quad \textrm{a.s.}
\end{equation*}
By the optional sampling theorem, $M^i_t / S^i_t = \mathbf{E}_{\mathbf{Q}}[M^i_{\rho_t}/S^i_t | \mathscr{F}_t]$, so we find that
\begin{equation*}
\frac{1}{(1+\varepsilon)^4} \leq \frac{M^i_t}{S^i_t} \leq (1+\varepsilon)^4\quad \textrm{a.s.,}
\end{equation*}
and we can conclude the proof by adjusting $\varepsilon$ appropriately.
\end{proof}

\begin{rem}\label{difference}
The difference between the proof of Theorem \ref{CPS-existence} and the earlier proofs of existence of CPSs under stronger assumptions in the papers \cite{Gua1,KS} is indeed that we introduce the second term in the definition of the increment \eqref{increment-def}. This term is indispensable as it ensures that the support condition \eqref{c-1} of Lemma \ref{keylem} holds even when $S$ is merely sticky. Namely, under stickiness the condition \eqref{c-1} might otherwise fail as it is possible, for example, that
\begin{equation*}
 \{0\} \neq \supp\mathscr{L}(S_{\tau_n}-S_{\tau_{n-1}}| \mathscr{F}_{\tau_{n-1}}) \subset[0,\infty)^d  \quad \textrm{a.s.\ on $\{\tau_{n-1}<T\}$}
\end{equation*}
(like in the case of Example \ref{mono}).
\end{rem}

\begin{rem}\label{alternative}
The probability measure $\mathbf{Q}$ in the proof of Theorem \ref{CPS-existence} can be alternatively constructed using Lemma 3.1 of \cite{Gua1}, instead of applying Lemma \ref{keylem}. More precisely, Lemma 3.1 of \cite{Gua1} can be applied sequentially to
\begin{equation*}
X \eqdefl \xi_{n}, \quad A \eqdefl \big\{\overline{S}_{\tau_{n-1}}=0\big\}, \quad \mathscr{G} \eqdefl \mathscr{F}_{\tau_{n-1}}, \quad \mathscr{H} \eqdefl \mathscr{F}_{\tau_{n}}, \quad \eta \eqdefl \frac{1}{2^n},
\end{equation*}
for any $n \in \mathbb{N}$, since \eqref{atom} and \eqref{zero-possible} hold, and since the inclusion \eqref{all-directions} ensures that, in fact,
\begin{equation*}
0 \in \mathrm{int}\,\mathrm{conv}\,\mathrm{supp}\,\mathscr{L}(\xi_n|\mathscr{F}_{\tau_{n-1}})\quad \textrm{a.s.\ on $\big\{\overline{S}_{\tau_{n-1}}>0\big\}$.}
\end{equation*}
As in \cite[p.\ 509]{Gua1}, the probability measure $\mathbf{Q}$ can then be defined through the density $L = \prod_{n=1}^\infty Z_n$, where $Z_n$, $n \in \mathbb{N}$, are the conditional densities obtained from Lemma 3.1 of \cite{Gua1}. The required properties of $\mathbf{Q}$ can be verified following \cite[pp.\ 509--510]{Gua1}. To this end, it is useful to note that
\begin{equation*}
\mathbf{1}_{\{\xi_{n-1}=0\}} \leq \mathbf{1}_{\big\{\overline{S}_{\tau_{n-1}}=0\big\}} \leq \mathbf{1}_{\{\xi_{n}=0\}}\quad \textrm{a.s.\ for any $n \in \mathbb{N}$.}
\end{equation*}
\end{rem}

It remains to prove the following lemma, used above in the proof of Theorem \ref{CPS-existence}, which a consequence of the stickiness of the process $S$.

\begin{lem}\label{prob}
For any $n \in \mathbb{N}$ and $i \in \{1,\ldots,2d+1\}$,
\begin{equation*}
\mathbf{P}[C_{n}^i|\mathscr{F}_{\tau_{n-1}}] >0 \quad \textrm{a.s.\ on $\big\{\overline{S}_{\tau_{n-1}}>0\big\}$,}
\end{equation*}
where $C_n^i$ and $\overline{S}_{\tau_{n-1}}$ are given by \eqref{cdef} and \eqref{support-maximum}, respectively.
\end{lem}
\begin{proof}
Fix $n \in \mathbb{N}$ and $i \in \{1,\ldots,2d+1\}$. Again to exclude a triviality (cf.\ the proof of Lemma \ref{keylem}), we assume that $\mathbf{P}\big[\overline{S}_{\tau_{n-1}}>0\big]>0$. Define a stopping time
\begin{equation*}
\rho \eqdefl \inf \bigg\{ t > \tau_{n-1} : \|S_t-S_{\tau_{n-1}} \| = \frac{2i-1}{2(2d+1)}\delta_{n-1}\bigg\}\wedge T
\end{equation*}
and note that $\mathbf{P}[\rho<T|\mathscr{F}_{\tau_{n-1}}]>0$ a.s.\ on $\big\{\overline{S}_{\tau_{n-1}}>0\big\}$ by the definition of $\overline{S}_{\tau_{n-1}}$ and $\delta_{n-1}$. Clearly, we have the inclusion
\begin{equation}\label{subset}
\{\rho<T\} \cap E  \subset C_{n}^{i},
\end{equation}
where
\begin{equation*}
E \eqdefl \bigg\{ S^\star_\rho < \frac{1}{2(2d+1)}\delta_{n-1}\bigg\}.
\end{equation*}
Let now $A \in \mathscr{F}_{\tau_{n-1}}$ be such that $\mathbf{P}[A]>0$ and $A \subset \big\{\overline{S}_{\tau_{n-1}}>0\big\}$. Observe that
\begin{equation*}
\mathbf{E}\big[\mathbf{1}_A\mathbf{P}[\{\rho<T\} \cap E|\mathscr{F}_{\tau_{n-1}}]\big] = \mathbf{E}\big[\mathbf{1}_{A \cap \{\rho<T\}}\mathbf{P}[E|\mathscr{F}_\rho]\big]>0
\end{equation*}
since $\mathbf{P}[A \cap \{\rho<T\}]>0$ and $\mathbf{P}[E|\mathscr{F}_\rho]>0$ a.s.\ by the stickiness of $S$ and Lemma \ref{strong-stickiness}. The assertion follows now from the inclusion \eqref{subset}.
\end{proof}

\section*{Acknowledgements}

M. S. Pakkanen and H. Sayit thank the Department of Mathematics at Saarland University for hospitality.
Additionally, M. S. Pakkanen acknowledges support 
from CREATES (DNRF78), funded by the Danish National Research Foundation, 
from the Aarhus University Research Foundation regarding the  project ``Stochastic and Econometric Analysis of Commodity Markets", and
from the Academy of Finland (project 258042). C. Bender acknowledges support by the Deutsche Forschungsgemeinschaft under grant BE3933/4-1.

\bibliographystyle{plain}

\end{document}